\documentclass[submission,copyright,creativecommons]{eptcs}
\providecommand{\event}{TARK 2021} % Name of the event you are submitting to
\usepackage{breakurl}             % Not needed if you use pdflatex only.
\usepackage{underscore}           % Only needed if you use pdflatex.

\usepackage{graphics,graphicx,epsfig,float} 
\usepackage{tikz}
\usetikzlibrary{arrows}
\usepackage{url}
\usepackage{latexsym,mathrsfs}
\usepackage{amssymb,amsfonts,amsmath}

\newcommand{\imp}{\rightarrow}

\newcommand{\et}{\wedge}

\newcommand{\Dia}{\Diamond}
\newcommand{\dia}[1]{\langle #1 \rangle}

\newcommand{\M}{\widehat{K}}

\renewcommand{\phi}{\varphi}
\newcommand{\union}{\cup}

\newcommand{\inter}{\cap}

\newcommand{\bisim}{{\raisebox{.3ex}[0mm][0mm]{\ensuremath{\medspace \underline{\! \leftrightarrow\!}\medspace}}}}
\newcommand{\bisrel}{\ensuremath{\mathfrak{R}}}

\newcommand{\weg}[1]{}

\newcommand{\lbr}{[\![}
\newcommand{\rbr}{]\!]}
 
\newcommand{\II}[1]{\lbr #1 \rbr} % consistent with Kluwer stylefile

\newcommand{\Formulas}{{\mathcal L}}

\newcommand{\langu}{\Formulas}

\newcommand{\state}{s}

\newcommand{\stateb}{t}

\newcommand{\Atoms}{P}

\newcommand{\atom}{p}

\newcommand{\Agents}{A}

\newcommand{\agent}{a}

\newcommand{\Domain}{{\mathcal D}}
\newcommand{\domain}{\Domain}

\newcommand{\Nat}{\mathbb N}
\newcommand{\Naturals}{\Nat}

\newcommand{\lang}{\langu}

\newcommand{\langpl}{\ensuremath{\lang_{\mathit{pl}}}}
\newcommand{\langel}{\ensuremath{\lang_{\mathit{el}}}}

\newcommand{\langpal}{\ensuremath{\lang_{\mathit{pal}}}}
\newcommand{\langapal}{\ensuremath{\lang_{\mathit{apal}}}}
\newcommand{\langbapal}{\ensuremath{\lang_{\mathit{bapal}}}}

\newcommand{\langgal}{\ensuremath{\lang_{\mathit{gal}}}}
\newcommand{\langcal}{\ensuremath{\lang_{\mathit{cal}}}}

\newcommand{\logicel}{\ensuremath{\mathit{EL}}}
\newcommand{\logicpal}{\ensuremath{\mathit{PAL}}}
\newcommand{\logicapal}{\ensuremath{\mathit{APAL}}}
\newcommand{\logicbapal}{\ensuremath{\mathit{BAPAL}}}

\newcommand{\logicgal}{\ensuremath{\mathit{GAL}}}
\newcommand{\logiccal}{\ensuremath{\mathit{CAL}}}
\newcommand{\logicqpal}{\ensuremath{\mathit{QPAL}}}

\newcommand{\know}{K}
\newcommand{\susp}{\M}

\newcommand{\var}{\mathit{var}}

\usepackage[thmmarks]{ntheorem}
\theoremsymbol{\ensuremath{\dashv}}
\usepackage{newproof}  % to have qed symbol, newproof must be after ntheorem...

\newtheorem{theorem}{Theorem}
\newtheorem{example}[theorem]{Example}
\newtheorem{definition}[theorem]{Definition}

\newtheorem{corollary}[theorem]{Corollary}

\usepackage{rotating}

\newcommand{\ol}[1]{\ensuremath{\overline{#1}}}

\newcommand{\fmp}{\ensuremath{\textbf{fmp}}}

\title{No Finite Model Property for Logics of Quantified Announcements}
\author{Hans van Ditmarsch
\institute{Open University of the Netherlands\\ Heerlen, The Netherlands}
\email{hans.vanditmarsch@ou.nl}
\and
Tim French
\institute{University of Western Australia\\
Perth, Australia}
\email{tim.french@uwa.edu.au}
\and
Rustam Galimullin
\institute{University of Bergen \\ Bergen, Norway}
\email{rustam.galimullin@uib.no}
}
\def\titlerunning{No FMP for Logics of Quantified Announcements}
\def\authorrunning{H. van Ditmarsch, T. French \& R. Galimullin}

\begin{document}
\maketitle

\begin{abstract}
Quantification over public announcements shifts the perspective from reasoning strictly about the results of a particular announcement to reasoning about the existence of an announcement that achieves some certain epistemic goal. Depending on the type of the quantification, we get different formalisms, the most known of which are arbitrary public announcement logic (\logicapal{}), group announcement logic (\logicgal), and coalition announcement logic (\logiccal). It has been an open question whether the logics have the finite model property, and in the paper we answer the question negatively. We also discuss how this result is connected to other open questions in the field.
\end{abstract}

\section{Introduction} \label{sec.one}
One of the most well-known ways to introduce communication into the multi-agent epistemic logic (\logicel{}) \cite{meyer95} is by extending the logic with public announcements, which, being dynamic operators, model the situation of all agents publicly and simultaneously receiving the same piece of information. Epistemic logic with such operators is called \emph{public announcement logic} (\logicpal{}) \cite{plaza07}, and it extends \logicel{} with constructs $[\varphi]\psi$ that mean `after public announcement of $\varphi$, $\psi$'. 

A natural way to generalise \logicpal{}, with epistemic planning \cite{bolander11} flavour to it, is to consider quantification over public announcements\footnote{See a recent survey \cite{vanditmarsch20} for an overview.}. Such an extension allows us to reason about the \emph{existence} of an announcement, or a sequence thereof, that reaches certain epistemic goal. There are several ways to quantify over public announcements, and we call the resulting logics \emph{quantified public announcement logics} (\logicqpal{}). In the paper, we consider the Big Three of \logicqpal{}'s, namely \emph{arbitrary public announcement logic} (\logicapal{}), \emph{group announcement logic} (\logicgal{}), and \emph{coalition announcement logic} (\logiccal{}). 

\logicapal{} \cite{balbiani08} extends \logicpal{} with constructs $\Box \varphi$ that are read as `after any public announcement, $\varphi$ is true'. While \logicapal{} modalities do not take into account who announces a formula or whether the formula can be truthfully announced by any of the agents in a system, \logicgal{} constructs $[G]\varphi$ restrict the quantification to the formulas that agents actually know \cite{agotnes10}. Thus, $[G]\varphi$ is read as `after any joint truthful announcement by agents from group $G$, $\varphi$ is true'. `Truthful' here denotes the fact that the agents know the formulas they announce. Finally, \logiccal{} is somewhat similar to \logicgal{} with a crucial difference that agents outside of the selected group can also make a simultaneous announcement \cite{agotnes08,thesis}. \logiccal{} extends \logicpal{} with constructs $[\!\langle G \rangle \! ] \varphi$ that are read as `whatever agents from $G$ announce, there is a simultaneous announcement by the agents outside of $G$, such that $\varphi$ is true after the joint announcement'. Given that the modalities of \logiccal{} were inspired by coalition logic \cite{pauly02}, it is not surprising that they are game-theoretic in nature, and, in particular, express the property of $\beta$-effectivity\footnote{The dual $\langle \! [ G ] \! \rangle \varphi$ expresses $\alpha$-effectivity.}.

One of the pressing open questions of \logicqpal{} is whether they have the finite model property (FMP).
\begin{quote}
FMP: \textit{A logic has the FMP iff every formula of the logic that is true in some model is also true in a finite model.}
\end{quote}
It is a standard result that \logicel{} has the FMP \cite{goranko07}. As the reader may have already guessed, after having read the title of the paper, we show that \logicapal{}, \logicgal{}, and \logiccal{} do not have the FMP.

The result is important for a couple of reasons. First, it tells us something about the expressivity of \logicqpal{}'s. In particular, all of \logicapal{}, \logicgal{}, and \logiccal{}, are so expressive that they can force infinite models, i.e. there are formulas of the languages that can only be true on infinite structures. 

Second, the lack of FMP sheds light on a connected open problem of finding finitary axiomatisations of \logicqpal{}'s. It is known \cite{urquhart81} that 
\begin{quote}
\textit{Finitary axiomatisation} and \textit{FMP} imply \textit{decidability}. (\textasteriskcentered)
\end{quote}

We also know that \logicapal{}, \logicgal{}, and \logiccal{} are all undecidable \cite{agotnes16}. The corresponding proof is quite complex, but ultimately it presents for each logic a formula over a parameterised set of tiles, such that the formula is satisfiable if and only if the set of tiles can tile the plane. The construction of the tiling is not explicit, in the sense that there is no one-to-one correspondence of worlds to unique points on the plane. Rather, given a model that satisfies the formula, a series of finite tilings is created in such a way that it guarantees the existence of an infinite tiling. It is not clear how to extract an argument against FMP from the undecidability proof, or whether it is possible. In any way, our approach presented in Section~\ref{sec:fmpapal} is simpler and more elegant. 

Another reason why we cannot get the lack of the FMP for free from the undecidability is that (\textasteriskcentered) requires the axiomatisations at hand to be finitary. To the best of our knowledge, none of \logicapal{}, \logicgal{}, and \logiccal{} have a known finitary axiomatisation, although the first two are recursively axiomatisable, and providing any axiomatisation of \logiccal{} is an open problem. However, (\textasteriskcentered) cannot be relaxed to requiring only recursive axiomatisations instead of finitary ones \cite{urquhart81,kracht91}. 

On the whole, the relation between the FMP, finitary axiomatisations, and decidability is not trivial, and (\textasteriskcentered) can be satisfied in various ways. For example, \logicel{} has all three properties, while there are modal logics that are decidable and recursively axiomatisable \cite{ognjanovic06,vanditmarsch20a}, or finitely axiomatisable, decidable, but lack the FMP \cite{gabbay71a,gabbay71b,bapal}, or not finitely axiomatisable, undecidable, but have the FMP \cite{gabbay98,kurucz00,hirsch02}, or have none of the three properties \cite{hirsch02}, and so on.

Due to the undecidability of \logicqpal{}'s, up until now there was a hope that if the logics have the FMP, then we will be able to conclude they are not finitely axiomatisable. We show that, in fact, the logics do not have the FMP and the problem of the existence of finitary axiomatisations remains.

In what follows, we formally introduce \logicapal{}, \logicgal{}, \logiccal{}, and some technical notions in Section \ref{sec.two}. In Section \ref{sec:fmpapal} we prove that \logicapal{} does not have the FMP. The strategy of the proof is such that we, first, present a formula that is true in an infinite model, and then claim that the formula cannot be true in any finite model. The results for \logicgal{} and \logiccal{} follow as a corollary. We conclude and discuss further research in Section \ref{sec:conc}.

\section{Quantified public announcement logics} \label{sec.two}

Given are a countable (finite or countably infinite) set of {\em agents} $\Agents$ and a countably infinite set of {\em propositional variables} $\Atoms$ (a.k.a.\ {\em atoms}, or {\em variables}). In what follows, $G \subseteq A$, and we denote $A \setminus G$ as $\overline{G}$.

%\paragraph*{Syntax} 
\subsection{Syntax} We start with defining the logical language and some crucial syntactic notions. 
\begin{definition}[Language] 
The language of quantified public announcement logic is defined as follows,  where $\agent\in\Agents$ and $\atom\in\Atoms$. 
\[ \lang(\Agents,\Atoms) \ \ni \ \phi ::= \atom \ | \ \neg \phi \ | \ (\phi \et \phi) \ | \ K_a \phi \ | \ [\phi]\phi \ | \ Q \phi \] 
\end{definition}
Other propositional connectives are defined by abbreviation, and, unless ambiguity results, we often omit parentheses occurring in formulas. We also often omit one or both of the parameters $\Agents$ and $\Atoms$ in $\lang(\Agents,\Atoms)$, and write $\lang(\Atoms)$ or $\lang$. Formulas are denoted $\phi,\psi$, possibly primed as in $\phi',\phi'',\dots, \psi',\dots$. Depending on which form the quantifier $Q$ takes --- $\Box$, $[G]$, or $[\!\langle G \rangle \!]$, where $G \subseteq \Agents$ --- we distinguish $\langapal$, $\langgal$, and $\langcal$, respectively. We also distinguish the language $\langel$ of \emph{epistemic logic} (without the constructs $[\phi]\phi$ and $Q\phi$).

For $K_a \phi$ read `agent $a$ knows $\phi$'. For $[\phi]\psi$, read `after public announcement of $\phi$, $\psi$'.  For $\Box \phi$, read `after any announcement, $\phi$ (is true)'. For $[G]\varphi$, read `after any joint announcement by agents from $G$, $\varphi$ is true'. And for $[\!\langle G \rangle \!]\varphi$ read `for each announcement by agents from $G$, there is a counter-announcement by the remaining agents, such that $\varphi$ is true after the joint simultaneous announcement'. The dual modalities are defined by abbreviation: $\susp_a\phi := \neg K_a \neg \phi$, $\dia{\phi}\psi := \neg [\phi] \neg \psi$, $\Dia\phi := \neg \Box \neg \phi$, $\langle G \rangle \varphi := \lnot [G]\lnot \varphi$, and $\langle \! [G] \! \rangle \varphi := \lnot [ \! \langle G \rangle \! ] \lnot \varphi$. 
%Unless ambiguity results we often omit one or both of the parameters $\Agents$ and $\Atoms$ in $\lang(\Agents,\Atoms)$, and write $\lang(\Atoms)$ or $\lang$. 
%Unless ambiguity results we often omit parentheses occurring in formulas. Formulas are denoted $\phi,\psi$, possibly primed as in $\phi',\phi'',\dots, \psi',\dots$

%We also distinguish the language $\langel$ of \emph{epistemic logic} (without the constructs $[\phi]\phi$ and $Q\phi$) and the language $\langpl$ of \emph{propositional logic} (without additionally the construct $K_a \phi$), also known as the \emph{Booleans}. Booleans are denoted $\phi_0, \psi_0$, etc. 

The set of propositional variables that occur in a given formula $\phi$ is denoted $\var(\phi)$ (where one that does not occur in $\phi$ is called a {\em fresh variable}), its {\em modal depth} $d(\phi)$ is the maximum nesting of $K_a$ modalities, and its {\em quantifier depth} $D(\phi)$ is the maximum nesting of $Q \in \{\Box, [G], [ \! \langle G \rangle \! ]\}$ modalities. These notions are inductively defined as follows.
\begin{itemize}
\item $\var(p) = \{p\}$, $\var(\neg\phi) = \var(K_a \phi) = \var(Q \phi) = \var(\phi)$, $\var(\phi\et\psi) = \var([\phi]\psi) = \var(\phi) \union \var(\psi)$; 

\item $D(p) = 0$, $D(\neg\phi) = D(K_a \phi) = D(\phi)$, $D(\phi\et\psi) = D([\phi]\psi) = \max \{D(\phi), D(\psi)\}$, $D(Q \phi) = D(\phi)+1$; 

\item $d(p) = 0$, $d(\neg\phi) = d(Q \phi) = d(\phi)$, $d(\phi\et\psi) = \max \{d(\phi), d(\psi)\}$, $d([\phi]\psi) = d(\phi)+d(\psi)$, $d(K_a \phi) = d(\phi)+1$.
\end{itemize}

\subsection{Structures} 
%\paragraph*{Structures} 
We consider the following structures and structural notions in this work.
\begin{definition}[Model]
An {\em (epistemic) model} $M = (S, \sim, V )$ consists of a non-empty {\em domain} $S$ (or $\domain(M)$) of {\em states} (or `worlds'), an {\em accessibility function} $\sim: \Agents \imp {\mathcal P}(S \times S)$, where each $\sim_\agent$ is an equivalence relation, and a {\em valuation} $V: \Atoms \imp {\mathcal P}(S)$, where each $V(\atom)$ represents the set of states where $\atom$ is true. For $\state \in S$, a pair $(M,s)$, for which we write $M_s$, is a {\em pointed (epistemic) model}. \end{definition} 
We will abuse the language and also call $M_s$ a model. We will occasionally use the following disambiguating notation: if $M$ is a model, $S^M$ is its domain, $\sim^M_a$ the accessibility relation for an agent $a$, and $V^M$ its valuation.

\begin{definition}[Bisimulation]\label{bisimulation}
    Let $M = (S,\sim,V)$ and $M' = (S',\sim',V')$
    be epistemic models. 
    A non-empty relation $\bisrel \subseteq S \times S'$
    is a {\em bisimulation} if for every 
    $(s, s') \in \bisrel$,
    $p \in P$, and
    $a \in A$ 
    the conditions {\bf atoms}, {\bf forth} and {\bf back} hold.
\begin{itemize}
\item    {\bf atoms}: 
    $s \in V(p)$ iff $s' \in V'(p)$.

\item    {\bf forth}: 
    for every $t \sim_{a} s$ 
    there exists $t' \sim'_{a} s'$
    such that $(t, t') \in \bisrel$.

\item    {\bf back}: 
    for every $t' \sim'_{a} s'$
    there exists $t \sim_{a} s$ 
    such that $(t, t') \in \bisrel$.
\end{itemize}
If there exists a bisimulation $\bisrel$ between $M$ and $M'$ such that $(s, s') \in \bisrel$, then 
    $M_{s}$ and $M'_{s'}$ are
    {\em bisimilar}, notation  
    $M_{s} \bisim M'_{s'}$ (or $\bisrel: M_{s} \bisim M'_{s'}$, to be explicit about the bisimulation).

Let $Q \subseteq P$. A relation $\bisrel$ between $M$ and $M'$ satisfying {\bf atoms} for all $p \in Q$, and {\bf forth} and {\bf back}, is a {\em $Q$-bisimulation} (a bisimulation {\em restricted} to $Q$). The notation for $Q$-restricted bisimilarity is $\bisim_Q$.
\end{definition}

\noindent The notion of $n$-bisimulation, for $n \in \Naturals$, is given by defining  relations $\bisrel^0 \supseteq \dots \supseteq \bisrel^n$.
\begin{definition}[$n$-Bisimulation] \label{n-bisimulation}
    Let $M = (S,\sim,V)$ and $M' = (S',\sim',V')$
    be epistemic models, and let $n \in \Naturals$. 
    A non-empty relation $\bisrel^0 \subseteq S \times S'$ 
    is a {\em $0$-bisimulation} if {\bf atoms} holds for pair $(s,s') \in\bisrel$. Then, a non-empty relation $\bisrel^{n+1} \subseteq S \times S'$ 
    is a {\em $(n+1)$-bisimulation} if for all $p \in P$ and $a \in A$:
\begin{itemize}
%\item    {\bf atoms}: 
 %   $s \in V(p)$ iff $s' \in V'(p)$;
%
\item    {\bf $(n+1)$-forth}: 
    for every $t \sim_{a} s$ 
    there exists $t' \sim'_{a} s'$
    such that $(t,t') \in \bisrel^n$;
\item    {\bf $(n+1)$-back}: 
    for every $t' \sim'_a s'$
    there exists $t \sim_{a} s$ 
    such that $(t,t') \in \bisrel^n$.
\end{itemize}
Similarly to $Q$-bisimulations we define {\em $Q$-$n$-bisimulations}, wherein {\bf atoms} is only required for $p \in Q \subseteq P$; $n$-bisimilarity is denoted $M_s \bisim^n M'_{s'}$, and $Q$-$n$-bisimilarity is denoted $M_s \bisim^n_Q M'_{s'}$.
\end{definition}

\subsection{Semantics}
%\paragraph*{Semantics}
We continue with the semantics of our logic. Let $\langel^G := \{\bigwedge_{i \in G} K_i \varphi_i \mid \varphi_i \in \langel\}$ be the set of all announcement by group $G$. 
\begin{definition}[Semantics] \label{def.truthlyingpub}
The interpretation of formulas in $\langapal \cup \langgal \cup \langcal$ on epistemic models is defined by induction on formulas.

Assume an epistemic model $M = (S, \sim, V )$, and $s \in S$.  
\[ \begin{array}{lcl}
M_s \models \atom &\mbox{iff} & \state \in V(p) \\ 
M_s \models \neg \phi &\mbox{iff} & M_s \not \models \phi \\ 
M_s \models \phi \et \psi &\mbox{iff} & M_s \models \phi  \text{ and } M_s \models \psi \\  
M_s \models K_\agent \phi &\mbox{iff} & \mbox{for all }  \stateb \in S: \state \sim_a \stateb \text{ implies } M_t  \models \phi \\  
M_s \models [\phi] \psi &\mbox{iff} & M_s \models \phi \text{ implies } M^\phi_s \models \psi \\
M_s \models \Box \psi & \mbox{iff} & \mbox{for all } \phi \in \langel : M_s \models [\phi] \psi\\
M_s \models [G] \psi &\mbox{iff} &\mbox{for all } \varphi_G \in \langel^G: M_s \models [\varphi_G] \psi\\
M_s \models [ \! \langle G \rangle \! ] \psi &\mbox{iff} &\mbox{for all }\varphi_G \in \langel^G \mbox{ there is } \chi_{\overline{G}} \in \langel^{\overline{G}}: M_s \models \varphi_G \rightarrow \langle\varphi_G \land \chi_{\overline{G}} \rangle \psi
\end{array} \] 
where $\II{\phi}_M := \{ s\in S \mid M_s\models \phi\}$; and where epistemic model $M^\phi = (S', \sim', V')$ is such that: $S' = \II{\phi}_M$, ${\sim'_a} = {\sim_a} \inter (\II{\phi}_M \times \II{\phi}_M)$, and $V'(p) := V(p) \inter \II{\phi}_M$. For $(M^\phi)^\psi$ we may write $M^{\phi\psi}$. Formula $\phi$ is {\em valid on model $M$}, notation $M \models \phi$, if for all $s \in S$, $M_s \models \phi$. Formula $\phi$ is {\em valid}, notation $\models \phi$, if for all $M$, $M \models \phi$. We call $\varphi$ \emph{satisfiable} if there is $M_s$ such that $M_s \models \varphi$.
\end{definition} 

Observe that the quantification in the definition of semantics is restricted to the quantifier-free fragment. Moreover, given the eliminability of public announcements from that fragment \cite{plaza07}, this amounts to quantifying over formulas of epistemic logic.

For clarity, we also give the semantics of the diamond versions of public and quantified announcements.

\[ \begin{array}{lcl}
M_s \models \langle \phi \rangle \psi &\mbox{iff} & M_s \models \phi \text{ and } M^\phi_s \models \psi \\
M_s \models \Dia \psi & \mbox{iff} & \mbox{there is } \phi \in \langel : M_s \models \langle \phi \rangle \psi\\
M_s \models [G] \psi &\mbox{iff} &\mbox{there is } \varphi_G \in \langel^G: M_s \models \langle \varphi_G \rangle \psi\\
M_s \models \langle \! [ G ] \! \rangle \psi &\mbox{iff} &\mbox{there is }\varphi_G \in \langel^G \mbox{ such that for all } \chi_{\overline{G}} \in \langel^{\overline{G}}: M_s \models \varphi_G \land [\varphi_G \land \chi_{\overline{G}} ] \psi
\end{array} \] 

%The logic $\langbapal$ results when we restrict the formula $\phi$ in the semantic definition of $\Box\psi$ to be a Boolean formula.

%The logic (set of validities) is called $\logicbapal$. 
\begin{definition}[Modal Equivalence]
Given $M_s$ and $M'_{s'}$, if for all $\phi \in\lang$, $M_s \models \phi$ iff $M'_{s'} \models \phi$, we write $M_{s} \equiv M'_{s'}$. Similarly, if this holds for all $\phi$ with $d(\phi) \leq n$, we write $M_{s} \equiv^n M'_{s'}$, and if this holds for all $\phi$ with $\var(\phi) \in Q \subseteq P$, we write $M_s \equiv_Q M'_{s'}$.  
\end{definition}

It is a standard standard model-theoretic result for $\langel$ that bisimulation between models implies their modal equivalence \cite{goranko07}. We can extend the result to $\langapal$, $\langgal$, and $\langcal$.

\begin{theorem}
Let $M_s$ and $M^\prime_{s^\prime}$ be models. Then $M_s \bisim M^\prime_{s^\prime}$ implies $M_{s} \equiv M'_{s'}$.
\end{theorem}

\begin{proof}
By an induction on the structure of a formula. The proof for the case of public announcements can be found, for example, in \cite{bapal}. The cases of quantified announcements follow by the induction hypothesis in the presence of an appropriate ordering of subformulas such that the $D$-depth takes precedence over $d$-depth.
\end{proof}

Another well-known result is that for $\langel$ (and hence for $\langpal$ given the translation from $\langpal$ into $\langel$ \cite{plaza07}), $M_s \bisim^n M'_{s'}$ implies $M_{s} \equiv^n M'_{s'}$ \cite{goranko07}.
 Observe that this is not the case for any of the \logicqpal{}'s because the quantification is over formulas of arbitrary finite modal depth and thus may exceed the given $n$. Finally, due to the fact that the quantification in the presented languages is \emph{implicit}, it is also not the case for \logicqpal{}'s that $M_s \bisim_Q M'_{s'}$ implies $M_s \equiv_Q M'_{s'}$.
 
 \begin{example}
 In order to highlight the differences between different types of quantification, let us consider models $M_{s_0}$, $M_{s_0}^\psi$, and $M_{s_0}^{\psi_a}$ in Figure \ref{fig:example}.
 \begin{figure}[h]
    \begin{center}
      \begin{tikzpicture}[>=stealth',shorten >=1pt,auto,node distance=5em,thick]
        \node (T0) {$t_0:p,\ol{q}$};
        \node (T1) [below of = T0] {$t_1:\ol{p}, \ol{q}$};
      
        \node (S0) [left of = T0] {$s_0:p,q$};
        \node (S1) [left of = T1] {$s_1:\ol{p},q$};
       
        \draw[dashed] (T0) edge node {$a$} (S0);
\draw[dashed] (T1) edge node {$a$} (S1);
   \draw (T0) edge node {$b$} (T1);
        \draw (S0) edge node {$b$} (S1);

      \end{tikzpicture} \hspace{3em}
       \begin{tikzpicture}[>=stealth',shorten >=1pt,auto,node distance=5em,thick]
        \node (T0) {$t_0:p,\ol{q}$};
        \node (T1) [below of = T0] {$t_1:\ol{p}, \ol{q}$};
      
        \node (S0) [left of = T0] {$s_0:p,q$};
       
        \draw[dashed] (T0) edge node {$a$} (S0);
   \draw (T0) edge node {$b$} (T1);

      \end{tikzpicture}
      \hspace{3em}
       \begin{tikzpicture}[>=stealth',shorten >=1pt,auto,node distance=5em,thick]
        \node (T0) {$t_0:p,\ol{q}$};
                \node (T1) [below of = T0] {};
      
        \node (S0) [left of = T0] {$s_0:p,q$};
       
        \draw[dashed] (T0) edge node {$a$} (S0);

      \end{tikzpicture}

    \end{center}
          \caption{Models, from left to right, $M_{s_0}$, $M_{s_0}^\psi$, and $M_{s_0}^{\psi_a}$. The names of the worlds indicate which atoms are true there.}
                \label{fig:example}
  \end{figure}
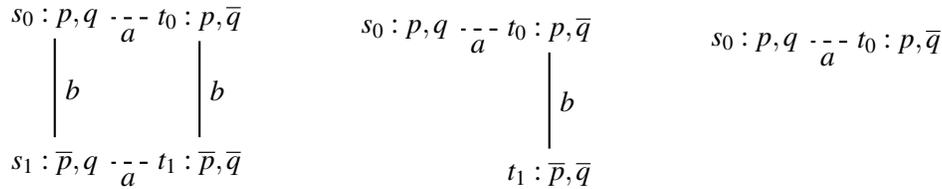
  There is a formula $\psi$ that can be announced in $M_{s_0}$, and such that $\varphi:= \susp_a \know_b p \land \susp_a \susp_b \lnot p$ will hold after the announcement. Indeed, let $\psi := \lnot p \imp \lnot q$. The result of updating $M_{s_0}$ with $\psi$ is presented in the figure and the reader can verify that $M_{s_0}^{\psi} \models \varphi$. This means that $M_{s_0} \models \langle \psi \rangle \varphi$, and hence $M_{s_0} \models \Dia \varphi$. Observe that $M_{s_0} \not \models \langle \{a\} \rangle \varphi$, since, according to the semantics, each announcement by agent $a$ should be prefixed with $\know_a$. This implies that in order to remove $s_1$, we also have to remove all $a$-reachable states, in particular $t_1$. On the other hand, $M_{s_0} \models \langle \{a\} \rangle (\know_b q \land \lnot \know_a q)$. Indeed, let $\psi_a := \know_a p$. Such an announcement results in model $M_{s_0}^{\psi_a}$ in which $\know_b q \land \lnot \know_a q$ is true. Finally, $M_{s_0} \not \models \langle \! [ \{a\} ] \! \rangle (\know_b q \land \lnot \know_a q)$, as any announcement by agent $a$ that results in a model with worlds $s_0$ and $t_0$ can be countered by agent $b$ with a simultaneous announcement $\know_b q$. Such a joint announcement results in a singleton model with the only world $s_0$.
 \end{example}

\section{APAL, GAL, and CAL do not have the finite model property}
\label{sec:fmpapal}
%Having shown that \logicbapal{} does not have the finite model property, we turn our attention to $\logicapal{}$. Here we will use the same strategy as was applied to \logicbapal{}, with some significant changes.

In this section, we show that none of the \logicqpal{}'s have the FMP. We do this by proving the result for \logicapal{} first, and then state the corresponding results for \logicgal{} and \logiccal{} as a corollary.

\begin{definition}[Finite Model Property]
A logic has the finite model property if every satisfiable formula is satisfied in a finite model.
\end{definition}

%The proof that $\logicbapal{}$ lacks the finite model property showed that there was a partition of a subset of worlds into type $A$ and type $B$ worlds, where for every type A world there is some Boolean announcement that preserves that world and only type $A$ worlds. However, there is no Boolean announcement that preserves only type $B$ worlds. So it followed that if we had a finite model, there would be a finite number of the Boolean announcements that preserved type $A$ worlds (up to bisimilarity), and the conjunction of their negations would be sufficient to provide an announcement that preserves only type $B$ worlds, so by contradiction it follows that there is no such finite model.

%Translating this proof to the case for \logicapal{} is complex. In the proof for $\logicbapal{}$ we could not suppose an announcement ``this is a type $B$'' world, because the property of being type $B$ was a modal property and the quantifiers ranged over Boolean formulas only. 

%Rewrite here!!!
Thie idea of the proof is to show that there is a formula, \fmp, that expresses the property that:
\begin{quote}
There is some subset of worlds in the model %$X$ 
that can be partitioned into sets $X$ and $Y$, such that for every element of $x\in X$, 
there is some announcement, $\psi^x$, that preserves $x$ and no states from $Y$, but there is no announcement that preserves only the states from $Y$, and none of the states from $X$.
\end{quote}
As the announcements in $\logicqpal$'s are closed under negations and conjunctions, if $X$ were finite up to bisimulation, then the announcement $\bigwedge_{x\in X}\lnot\psi^x$, 
would be adequate to preserve only the states of $Y$ and none of the states from $X$. 
Therefore, if it can be shown that such a formula is satisfiable, then it would follow that $\logicqpal$'s do not have the finite model property.
To show the satisfiability for such a formula, we need a means to identify $X$ and $Y$ states. 
An epistemic formula will not do, as any formula that characterises $X$, will have a negation that characterises $Y$, and announcing the negation would violate the property we need.
However, rather than distinguishing between two partitions using a modal property, we can distinguish them using a second order property that the $\logicapal{}$ quantifier does not range over. 
We know  \cite{agotnes16} that \logicapal{} can express such properties; 
particularly it can specify whether or not two states are $n$-bisimilar for all $n$. It is this property we use to define the necessary partition of worlds.

The proof will be via construction, where we will present a $\logicapal$ formula, show that it has an infinite model, and then show that it is impossible that a finite model exists.

\begin{theorem}\label{thm:afmp}
  \logicapal{} does not have the finite model property.
\end{theorem}
\newcommand{\Root}{\ensuremath{\mathbf{root}}}
\newcommand{\stem}{\ensuremath{\mathbf{stem}}}
\newcommand{\tier}{\ensuremath{\mathbf{tier}}}
\begin{proof}
  We will give the proof by construction. Consider the following formula $\fmp$.
  \begin{eqnarray*}
    \Root &=& \Box(\susp_a(\lnot x\land \know_b\lnot x)\imp\know_a(\lnot x\imp\know_b\lnot x))\\
    \stem &=& \Dia(\susp_a(\lnot x\land\know_b\lnot x)\land\susp_a(\lnot x\land\susp_b x))\\
    \tier &=& \know_b(x\land \susp_a\lnot x\land \know_a(\lnot x\imp\susp_b x))\\
    \fmp &=& \bigwedge\left(
      \begin{array}{l}
        \tier\land \susp_b\Root \land \susp_b\stem \\
        \know_b(\stem \imp \Dia(\tier\land\know_b\stem))\\ 
        \know_b(\Root\imp\Box(\tier\imp\susp_b\stem))
      \end{array}
    \right)
  \end{eqnarray*}
  The formula $\fmp$ partitions a set of $b$-related worlds into two sets: 
  {\em root} (the worlds where the formula $\Root$ is true);
  and {\em stem} (the worlds where the formula $\stem$ is true).
  We note that $\stem$ is equivalent to $\lnot\Root$ so this is a partition.

  The formula $\tier$ sets a label $x$, where $x$ is true at all the $b$-related worlds ({\em tier}-0), 
  at every $b$-related world there is an $a$-related world where $x$ is false ({\em tier}-1), 
  and from each of those worlds there is a $b$-related world where $x$ is true ({\em tier}-2). 
  This creates a consistent labelling used to define $\Root$ and $\stem$.

  The formula $\Root$ is true at a world if there is only one $a$-reachable tier-1 world, up to finite bisimulation,
  while $\stem$ is true if there is more than one $a$-reachable tier-1 world. 
  
  The first line of $\fmp$ states that $\tier$ is true, and at least one of the $b$-related worlds satisfies $\stem$, and at least one of the $b$-related worlds satisfies $\Root$. 
  The second line of $\fmp$ states that there is some public announcement that removes all of the $\Root$ worlds (and possibly some, but not all, of the $\stem$ worlds too),
  and the final line of $\fmp$ states that there is no public announcement that removes all of the $\stem$ worlds leaving a $\Root$ world. 
  This line needs a small caveat. 
  Rather than just talking about removing $\stem$ and $\Root$ worlds, 
  there is the possibility that an announcement could change a $\stem$ world to a $\Root$ world, or vice-versa.
  However each of these announcement quantifiers is guarded by the formula $\tier$, and we have the property $\Root\imp\Box\Root$.
  Therefore, we do not need to consider the cases where $\Root$ worlds are transformed into $\stem$ worlds.

As an example consider a finite model in Figure \ref{fig:finite}. The model does not satisfy \fmp, and in particular it does not satisfy the third conjunct. Indeed, let the current world be $s_0$ that satisfies $\Root$. We show that $M_{s_0} \not \models \Box(\tier\imp\susp_b\stem)$, or, equivalently, $M_{s_0} \models \Dia(\tier \land \know_b \lnot \stem)$. By the semantics this amounts to the fact that there is $\psi \in \langel$ such that $\tier \land \know_b \lnot \stem$ holds in the updated model. Let, for example, $\psi :=\know_a(\lnot x \imp \know_b (x \imp \know_a \lnot p_4))$. It is easy to verify that the resulting updated model, which only consists of states $s_0$, $t_0$, and $u_0$, satisfies $\tier \land \know_b \lnot \stem$.

 \begin{figure}[h]
    \begin{center}
      \begin{tikzpicture}[>=stealth',shorten >=1pt,auto,node distance=5em,thick]
        \node (T0) {$t_0:\ol{x}$};
        \node (T1) [below of = T0] {$t_1:\ol{x}$};
        \node (T2) [below of = T1] {$t_2:\ol{x}$};
        \node (T3) [below of = T2] {$t_3:\ol{x}$};
        \node (T4) [below of = T3] {$t_4:\ol{x}$};
        \node (S0) [left of = T0] {$s_0:x$};
        \node (S1) [left of = T1] {$s_1:x$};
        \node (S2) [left of = T3] {$s_2:x$};
        \node (U0) [right of = T0] {$u_0:x$};
        \node (U1) [right=50pt, right of = T1] {$u_1:x,p_1,p_2,p_3,p_4,p_5\hdots$};
        \node (U2) [right=50pt, right of = T2] {$u_2:x,p_2,p_3,p_4,p_5,p_6\hdots$};
        \node (U3) [right=50pt, right of = T3] {$u_3:x,p_3,p_4,p_5,p_6,p_7\hdots$};
        \node (U4) [right=50pt, right of = T4] {$u_4:x,p_4,p_5,p_6,p_7,p_8\hdots$};
        \draw[dashed] (T0) edge node {$a$} (S0);
        \draw[dashed] (T1) edge (T2);\draw[dashed] (T2) edge (S1);\draw[dashed] (T1) edge node {$a$} (S1);
        \draw[dashed] (T3) edge (T4);\draw[dashed] (T4) edge (S2);\draw[dashed] (T3) edge node {$a$} (S2);
        \draw (S0) edge node {$b$} (S1);
        \draw (S1) edge node {$b$} (S2);
        \draw (T0) edge node {$b$} (U0);
        \draw (T1) edge node {$b$} (U1);
        \draw (T2) edge node {$b$} (U2);
        \draw (T3) edge node {$b$} (U3);
        \draw (T4) edge node {$b$} (U4);
      \end{tikzpicture}
      \caption{A model that does not satisfy formula~\fmp. The names of the worlds indicate which atoms are true there, e.g. $x$ is true and $\ol{x}$ means that $x$ is false.}
      \label{fig:finite}
    \end{center}
  \end{figure}
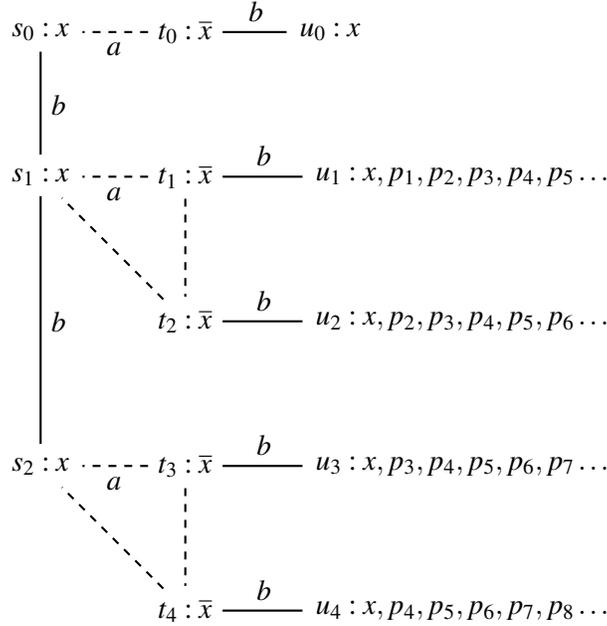

  Now we show that $\fmp$ is a satisfiable formula. 
%  The reasoning presented in Section~\ref{sbsect:bapal} is sufficient to show that $\fmp$ cannot have a finite model.
  The model $M=(S,\sim,V)$, which satisfies $\fmp$, is built as follows:
  \begin{enumerate}
    \item $S = \{s_0,s_1,\hdots\}\cup \{t_0,t_1,\hdots\}\cup \{u_0,u_1,\hdots\}$
    \item $\forall i\geq 0$, $s_i\sim_a s_i$, $s_i\sim_a t_{2i}$, $t_{2i}\sim_a s_i$, $t_i\sim_a t_i$ and $u_i\sim_a u_i$.
    \item $\forall i>0$, $s_i\sim_a t_{2i-1}$, $t_{2i}\sim_a t_{2i-1}$, $t_{2i-1}\sim_a s_i$ and $t_{2i-1}\sim_a t_{2i}$.
    \item $\forall i,j\geq 0,\ s_i\sim_b s_j$, $t_i\sim_b u_i$, $u_i\sim_b t_i$, $t_i\sim_b t_i$ and $u_i\sim_b u_i$.
    \item $V(x) = \{s_0,s_1,\hdots\}\cup \{u_0,u_1,\hdots\}$
    \item $\forall i\geq 0$, $V(p_i) = \{u_k\ |\ 0<k \leqslant i\}$.
  \end{enumerate}
  This model is represented in Figure~\ref{fig:rootandstem}. 
    \begin{figure}[H]
    \begin{center}
      \begin{tikzpicture}[>=stealth',shorten >=1pt,auto,node distance=5em,thick]
        \node (T0) {$t_0:\ol{x}$};
        \node (T1) [below of = T0] {$t_1:\ol{x}$};
        \node (T2) [below of = T1] {$t_2:\ol{x}$};
        \node (T3) [below of = T2] {$t_3:\ol{x}$};
        \node (T4) [below of = T3] {$t_4:\ol{x}$};
        \node (T5) [below of = T4] {$t_5:\ol{x}$};
        \node (T6) [below of = T5] {$t_6:\ol{x}$};
        \node (T7) [below of = T6] {$t_7:\ol{x}$};
        \node (T8) [below of = T7] {$t_8:\ol{x}$};
        \node (S0) [left of = T0] {$s_0:x$};
        \node (S1) [left of = T1] {$s_1:x$};
        \node (S2) [left of = T3] {$s_2:x$};
        \node (S3) [left of = T5] {$s_3:x$};
        \node (S4) [left of = T7] {$s_4:x$};
        \node (U0) [right of = T0] {$u_0:x$};
        \node (U1) [right=50pt, right of = T1] {$u_1:x,p_1,p_2,p_3,p_4,p_5\hdots$};
        \node (U2) [right=50pt, right of = T2] {$u_2:x,p_2,p_3,p_4,p_5,p_6\hdots$};
        \node (U3) [right=50pt, right of = T3] {$u_3:x,p_3,p_4,p_5,p_6,p_7\hdots$};
        \node (U4) [right=50pt, right of = T4] {$u_4:x,p_4,p_5,p_6,p_7,p_8\hdots$};
        \node (U5) [right=50pt, right of = T5] {$u_5:x,p_5,p_6,p_7,p_8,p_9,\hdots$};
        \node (U6) [right=50pt, right of = T6] {$u_6:x,p_6,p_7,p_8,p_9,p_{10},\hdots$};
        \node (U7) [right=50pt, right of = T7] {$u_7:x,p_7,p_8,p_9,p_{10},p_{11},\hdots$};
        \node (U8) [right=50pt, right of = T8] {$u_8:x,p_8,p_9,p_{10},p_{11},p_{12},\hdots$};
        \node (etc) [below of = S4] {$\vdots$};
        \draw[dashed] (T0) edge node {$a$} (S0);
        \draw[dashed] (T1) edge (T2);\draw[dashed] (T2) edge (S1);\draw[dashed] (T1) edge node {$a$} (S1);
        \draw[dashed] (T3) edge (T4);\draw[dashed] (T4) edge (S2);\draw[dashed] (T3) edge node {$a$} (S2);
        \draw[dashed] (T5) edge (T6);\draw[dashed] (T6) edge (S3);\draw[dashed] (T5) edge node {$a$} (S3);
        \draw[dashed] (T7) edge (T8);\draw[dashed] (T8) edge (S4);\draw[dashed] (T7) edge node {$a$} (S4);
        \draw (S0) edge node {$b$} (S1);
        \draw (S1) edge node {$b$} (S2);
        \draw (S2) edge node {$b$} (S3);
        \draw (S3) edge node {$b$} (S4);
        \draw (S4) edge node {$b$} (etc);
        \draw (T0) edge node {$b$} (U0);
        \draw (T1) edge node {$b$} (U1);
        \draw (T2) edge node {$b$} (U2);
        \draw (T3) edge node {$b$} (U3);
        \draw (T4) edge node {$b$} (U4);
        \draw (T5) edge node {$b$} (U5);
        \draw (T6) edge node {$b$} (U6);
        \draw (T7) edge node {$b$} (U7);
        \draw (T8) edge node {$b$} (U8);
      \end{tikzpicture}
      \caption{A model that does satisfy formula~\fmp. For each world $s_i$, where $i>0$, there is an announcement that preserves $s_i$, but not $s_0$. 
      However, there is no announcement that preserves only $s_0$. \label{fig:rootandstem}}
    \end{center}
  \end{figure}
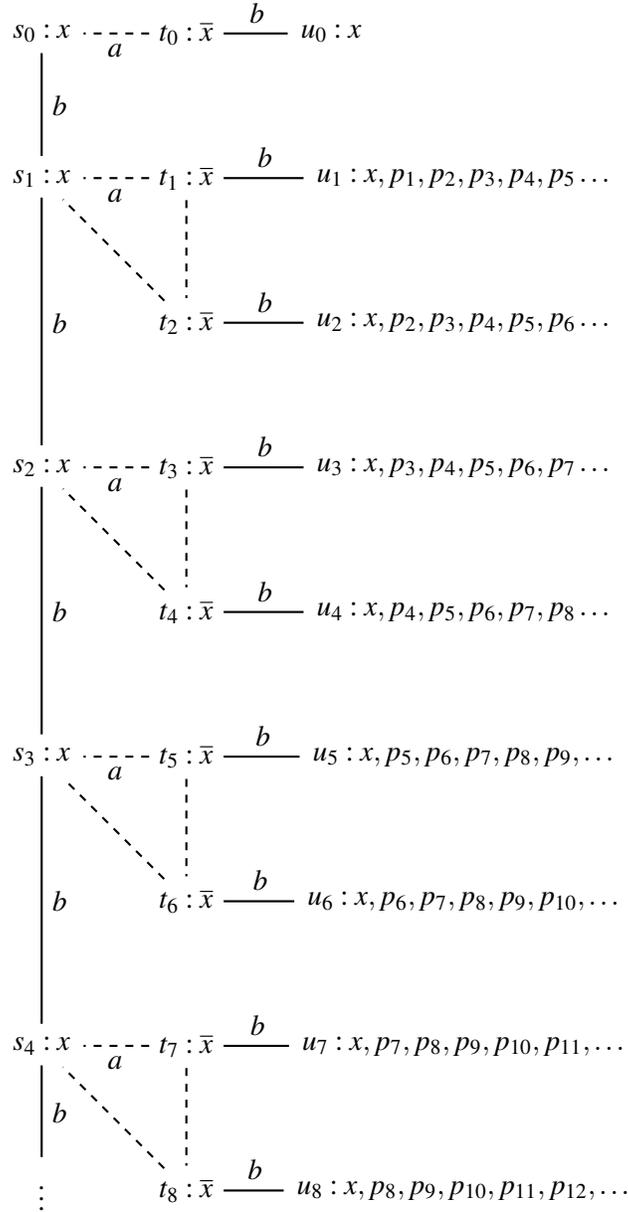
  The formula $\Root$ is true only at the state $s_0$, and the formula $\stem$ is true at $s_i$ for all $i>0$.
  The states $s_i$ are tier-0 states, the states $t_i$ are tier-1 states, and the states $u_i$ are tier-2 states.
  Therefore $M_{s_0}\models \tier\land\susp_b\Root\land\susp_b\stem$.
  In any state $s_i$ that satisfies $\stem$, we can make an announcement $\psi_i :=\susp_a\susp_b p_{2i}$ that will preserve 
  all the states $s_j, t_j,u_j$ where $0<j\leq 2i$. 
  Therefore $M_{s_i}\models \langle \psi_i \rangle (\tier\land\know_b\stem)$, so $M_{s_0}\models\know_b(\stem\imp\Dia(\tier\land\know_b\stem))$.
  Finally, consider any announcement, $\psi$ that preserves the root state $s_0$ and keeps $\tier$ true.
  Suppose that $\var(\psi)\subseteq P_n = \{x,p_0,\hdots,p_n\}$.
  The state $s_0$ is $P_n$-bisimilar to all states $s_i$ for $i>n$, so all these states will be preserved and continue to satisfy $\stem$.
  Therefore $M_{s_0}\models\know_b(\Root\imp\Box(\tier\imp\susp_b\stem))$ as required. Since all three conjuncts are now satisfied,  $M_{s_0}\models\fmp$.

  Finally, we reason that $\fmp$ cannot have a finite model via a contradiction. Suppose that $\fmp$ did have some finite model $M = (S,\sim,V)$. 
  We know $\Root$ is equivalent to $\lnot\stem$, and $\Root\imp\Box(\tier\imp\Root)$.
  As $M_s\models\fmp$ we have that for each $s_1,s_2,\hdots, s_k$, where $s\sim_b s_i$ and $M_{s_i}\models\stem$, 
  there is some announcement $\psi_i$ such that $M_{s_i}^{\psi_i}\models\tier\land\know_b\stem$. 
  Therefore, for some $t\sim_b s$ where $M_t\models\Root$, we have $M_t\models\lnot\psi_i$ for $i = 1,\hdots, k$.
  Announcing $\varphi:=\bigwedge_{i=1}^k(\tier\imp\lnot\psi_i)$ at $t$ therefore preserves $t$, and removes every $s_i\sim s$ where $M_{s_i}\not\models\Root$.
  So it follows that $M_t\models\Dia(\tier\land\know_b\Root)$, contradicting the fact that $M_s\models\fmp$. 

\end{proof}

The proof of Theorem \ref{thm:afmp} can be extended to the cases of \logicgal{} and \logiccal{}. To see this, it is enough to notice that the formula $\fmp$ forces a model with a root-and-stem structure, and thus arbitrary announcements can be `modelled' by announcements of the grand coalition. In particular, we need to substitute $\Box$ and $\Dia$ in $\fmp$ with $[A]$ and $\langle A \rangle$ for \logicgal{}, and with $[ \! \langle A \rangle \! ]$ and $\langle \! [A ] \! \rangle$ for \logiccal{}. Finally, the model in Figure \ref{fig:rootandstem} will also work, since the intersection of $a$- and $b$-relations is the identity, and the set of all possible \logicapal{} updates then coincides with those of \logicgal{} and \logiccal{}.

\begin{corollary}
\logicgal{} and \logiccal{} do not have the finite model property.
\end{corollary}

\section{Conclusions and further research}
\label{sec:conc}
It has been an open question for quite a while whether quantified public announcement logics have the finite model property, and we have answered the question for \logicapal{}, \logicgal{}, and \logiccal{} negatively. Not only this result is interesting in itself, it also clarifies some other properties of \logicqpal{}'s. In particular, from the expressivity perspective, we presented a formula that forces infinite models. Moreover, we have found the value of one of the unknowns in the expression
\begin{quote}
\textit{Finitary axiomatisation} and \textit{FMP} imply \textit{decidability}
\end{quote}
and thus only the problem of finding finitary axiomatisations of \logicqpal{}'s stands. Finally, the result trivially extends to the logics that are extensions of any of \logicqpal{}'s (e.g. \cite{galimullin21}).

Interestingly, restricting the quantification in \logicapal{} to just announcements of Boolean formulas still results in a logic with no FMP. This was shown in \cite{bapal}, where the authors used a somewhat simpler partition of worlds distinguishable by a formula of epistemic logic. Since in \logicapal{} the quantifiers range over all epistemic formulas, we used a more complex second-order property of worlds being $n$-bisimilar for arbitrary $n$. It is unknown whether \logicapal{} with the quantification restricted only to positive (universal) fragment of epistemic logic \cite{vanditmarsch20a} has the FMP, but given the aforementioned results, it seems very unlikely.

\section*{Acknowledgments}
We thank the anonymous reviewers of the paper.
This work was carried out while Hans van Ditmarsch was affiliated to LORIA, CNRS, University of Lorraine, France.

%\red{This shows that the both $\logicbapal{}$ and $\logicapal{}$ do not have the finite model property, and via the L\"owenheim-Skolem, 
%that there can be no finite first order axiomatization of either logic, giving a partially negative answer to the open problem posed by Kuijer when he showed the proposed axiomatization of \cite{balbiani08} was unsound. Of course, the lack of finite model property is not especially relevant here, as the proposed axiomatization of \cite{balbiani08} was an axiom schema, and thus never a finite first order axiomatization. Indeed, the lack of finite model property for \logicbapal{} was reported along side a sound and complete axiom schema for $\logicbapal{}$ \cite{bapal}. Nevertheless, this result still fills an important hole in the model theory for logics of quantified announcements. This result may easily be extended to the logics $\logicgal{}$ and $\logiccal{}$.

%This leaves open the question as to whether $\logicpapal$ has the finite model property, although given the results here, it seems very unlikely. The question as to whether $\logicapal{}$ has a sound and complete finite axiom schema also remains open.}

\bibliographystyle{eptcs}
\bibliography{references.bib}

\end{document}